\documentclass[11pt]{article}
\usepackage[utf8]{inputenc}
\usepackage{amsthm, amsmath, amssymb, empheq, tikz, physics, comment, graphicx, float, caption, subcaption, mathtools, cancel, bbm}

\usepackage[colorlinks=true, allcolors=blue]{hyperref}

\usepackage[margin=1in]{geometry}

\allowdisplaybreaks[1]

\newtheorem{theorem}{Theorem}%[section]

\newtheorem{lemma}{Lemma}

\newcommand{\p}{\mathbb{P}} % Probability P
\newcommand{\E}{\mathbb{E}} % Expectation E

\begin{document}

%%%%%%%%%%%%%
%%% Title %%%
%%%%%%%%%%%%%

\title{Tree trace reconstruction using subtraces}
\author{
	Tatiana Brailovskaya 
	\thanks{Princeton University; \url{tatianab@princeton.edu}.} 
	\and
	Mikl\'os Z.\ R\'acz
	\thanks{Princeton University; \url{mracz@princeton.edu}. Research supported in part by NSF grant DMS 1811724 and by a Princeton SEAS Innovation Award.} 
}
\date{\today}

\maketitle

%%%%%%%%%%%%%%%%
%%% Abstract %%%
%%%%%%%%%%%%%%%%

\begin{abstract}
    Tree trace reconstruction aims to learn the binary node labels of a tree, given independent samples of the tree passed through an appropriately defined deletion channel. In recent work, Davies, R\'acz, and Rashtchian~\cite{DRRAAP} used combinatorial methods to show that $\exp(\mathcal{O}(k \log_{k} n))$ samples suffice to reconstruct a complete $k$-ary tree with $n$ nodes with high probability. We provide an alternative proof of this result, which allows us to generalize it to a broader class of tree topologies and deletion models. In our proofs, we introduce the notion of a subtrace, which enables us to connect with and generalize recent mean-based complex analytic algorithms for string trace reconstruction. 
\end{abstract}

%%%%%%%%%%%%%%%%
%%% Document %%%
%%%%%%%%%%%%%%%%

%%%%%%%%%%%%%%%%%%%%%%%%%%%%%%%%%%%%%%%%%%%%
\section{Introduction} \label{sec:intro} %%%
%%%%%%%%%%%%%%%%%%%%%%%%%%%%%%%%%%%%%%%%%%%%

Trace reconstruction is a fundamental statistical reconstruction problem which has received much attention lately. Here the goal is to infer an unknown binary string of length $n$, given independent copies of the string passed through a deletion channel. The deletion channel deletes each bit in the string independently with probability $q$ and then concatenates the surviving bits into a \textit{trace}. The goal is to learn the original string with high probability using as few traces as possible. 

The trace reconstruction problem was introduced two decades ago~\cite{levenshtein2001efficient,BatuKannan04-RandomCase}, and despite lots of work over the past two decades~\cite{HolensteinMPW08,DeOdonnellServedio17-WorseCase,de2019optimal,NazarovPeres17-WorseCase,HartungHP18,holden2020subpolynomial,LB1,LB2,chase2020ub,chen2020polynomialtime,grigorescu2020limitations}, understanding the sample complexity of trace reconstruction remains wide open. 
Specifically, the best known upper bound is due to Chase~\cite{chase2020ub} who showed that $\exp(\widetilde{\mathcal{O}}(n^{1/5}))$ samples suffice; 
this work builds upon previous breakthroughs by De, O'Donnell, and Servedio~\cite{DeOdonnellServedio17-WorseCase,de2019optimal}, and Nazarov and Peres \cite{NazarovPeres17-WorseCase}, who simultaneously obtained an upper bound of $\exp(\mathcal{O}(n^{1/3}))$. In contrast, the best known lower bound is $\widetilde{\Omega}(n^{3/2})$ (see~\cite{LB1,LB2}). 
Considering average-case strings as opposed to worst-case ones reduces the sample complexity considerably, but the large gap remains: the current best known upper and lower bounds are $\exp(\mathcal{O}(\log^{1/3} n))$ (see~\cite{holden2020subpolynomial}) and $\widetilde{\Omega}(\log^{5/2}n)$ (see~\cite{LB2}), respectively. As we can see, the bounds are exponentially far apart for both the worst-case and average-case problems.

Given the difficulty of the trace reconstruction problem, several variants have been introduced, in part to study the strengths and weaknesses of various techniques. 
These include generalizing trace reconstruction from strings to trees~\cite{DRRAAP} and matrices~\cite{KMMP19}, coded trace reconstruction~\cite{CGMR,BLS20}, population recovery~\cite{Ban,BCSS19,narayanan2021population}, and more~\cite{narayanan2020circular,chen2020polynomialtime,davies2020approximate}. 

In this work we consider tree trace reconstruction, introduced recently by Davies, R\'acz, and Rashtchian~\cite{DRRAAP}. In this problem we aim to learn the binary node labels of a tree, given independent samples of the tree passed through an appropriately defined deletion channel. 
The additional tree structure makes reconstruction easier; indeed, in several settings Davies, R\'acz, and Rashtchian~\cite{DRRAAP} show that the sample complexity is polynomial in the number of bits in the worst case. Furthermore, Maranzatto~\cite{maranzatto_thesis} showed that strings are the hardest trees to reconstruct; that is, the sample complexity of reconstructing an arbitrary labeled tree with $n$ nodes is no more than the sample complexity of reconstructing an arbitrary labeled $n$-bit string.    

As demonstrated in~\cite{DRRAAP}, tree trace reconstruction provides a natural testbed for studying the interplay between combinatorial and complex analytic techniques that have been used to tackle the string variant. Our work continues in this spirit. In particular, Davies, R\'acz, and Rashtchian~\cite{DRRAAP} used combinatorial methods to show that $\exp(\mathcal{O}(k \log_{k} n))$ samples suffice to reconstruct complete $k$-ary trees with $n$ nodes, and here we provide an alternative proof using complex analytic techniques. This alternative proof also allows us to generalize the result to a broader class of tree topologies and deletion models. 
Before stating our results we first introduce the tree trace reconstruction problem more precisely. 

Let $X$ be a rooted tree with unknown binary labels on its $n$ non-root nodes. 
We assume that~$X$ has an ordering of its nodes, and the children of a given node have a left-to-right ordering. 
The goal of tree trace reconstruction is to learn the labels of $X$ with high probability, using as few traces as possible, knowing only the deletion model, the deletion probability $q < 1$, and the tree structure of~$X$. Throughout this paper, we write `with high probability' to mean with probability tending to~$1$ as~$n \to \infty$. 

While for strings there is a canonical model of the deletion channel, there is no such canonical model for trees. 
Previous work in~\cite{DRRAAP} considered two natural extensions of the string deletion channel to trees: the Tree Edit Distance (TED) deletion model and the Left-Propagation (LP) deletion model; see~\cite{DRRAAP} for details. 
Here we focus on the TED model, while also introducing a new deletion model, termed All-Or-Nothing (AON), which is more `destructive' than the other models. 
In both models the root never gets deleted. 

\begin{itemize}
    \item \textbf{Tree Edit Distance (TED) deletion model:}  
    Each non-root node is deleted independently with probability $q$ and deletions are associative. 
    When a node $v$ gets deleted, all of the children of~$v$ now become children of the parent of $v$. 
    Equivalently, contract the edge between $v$ and its parent, retaining the label of the parent. 
    The children of $v$ take the place of $v$ in the left-to-right order; in other words, the original siblings of $v$ that are to the left of $v$ and survive are now to the left of the children of $v$, and the same holds to the right of $v$. 
    
    \item \textbf{All-Or-Nothing (AON) deletion model:} 
    Each non-root node is marked independently with probability $q$. 
    If a node $v$ is marked, then the whole subtree rooted at $v$ is deleted. 
    In other words, a node is deleted if and only if it is marked or it has an ancestor which is marked. 
\end{itemize}

Figure~\ref{fig:del_models} illustrates these two deletion models. 
We refer to~\cite{DRRAAP} for motivation and further remarks on the TED deletion model. While the AON deletion model is significantly more destructive than the TED deletion model, 
an advantage of the tools we develop in this work is that we are able to obtain similar results for arbitrary tree topologies under the AON deletion model. 

\begin{figure}[t]
     \centering
     \begin{subfigure}[b]{0.31\textwidth}
         \centering
         \includegraphics[width=\textwidth]{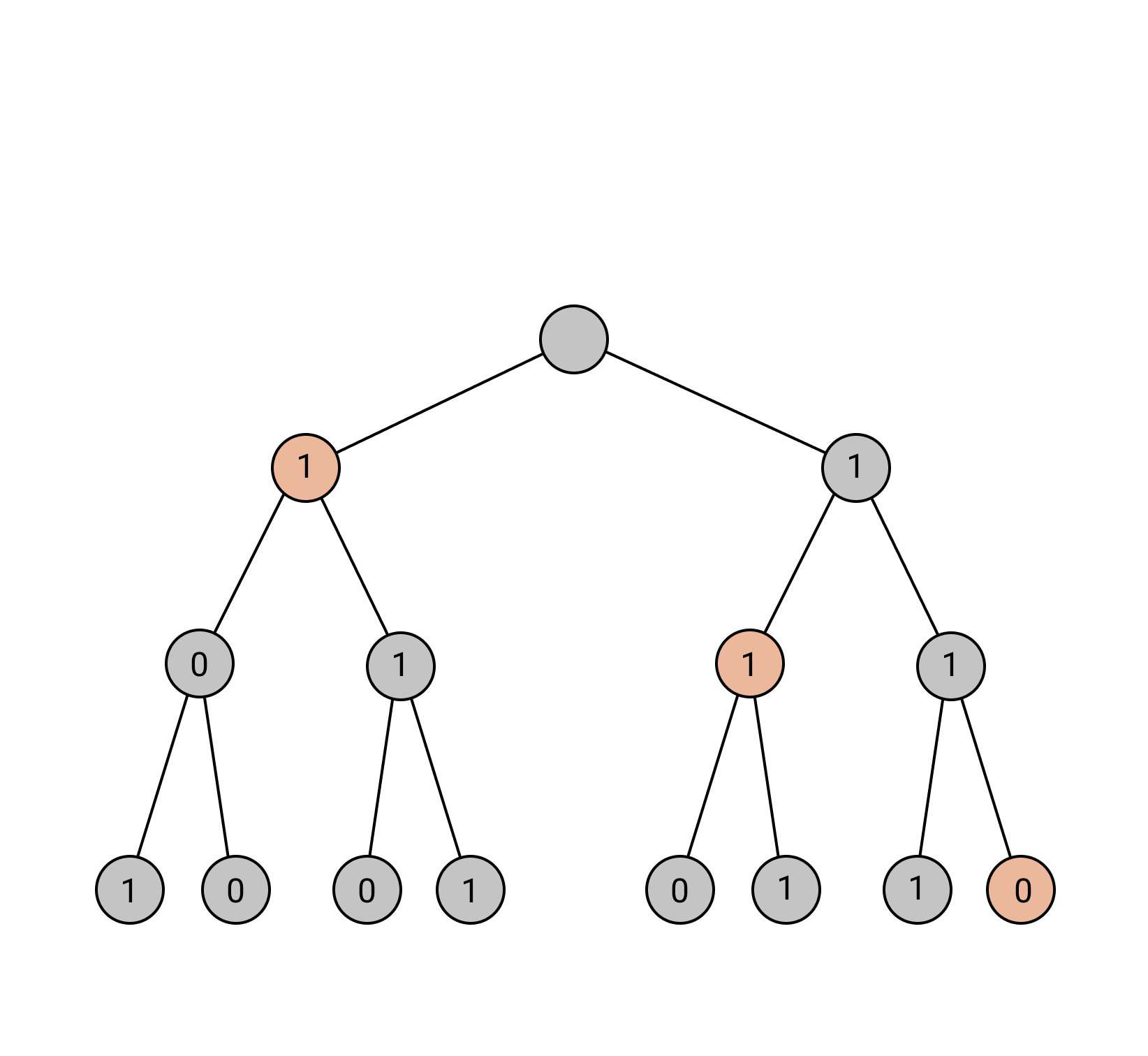}
         \caption{Original tree}
         \label{fig:orig}
     \end{subfigure}
     %\hfill
     \begin{subfigure}[b]{0.31\textwidth}
         \centering
         \includegraphics[width=\textwidth]{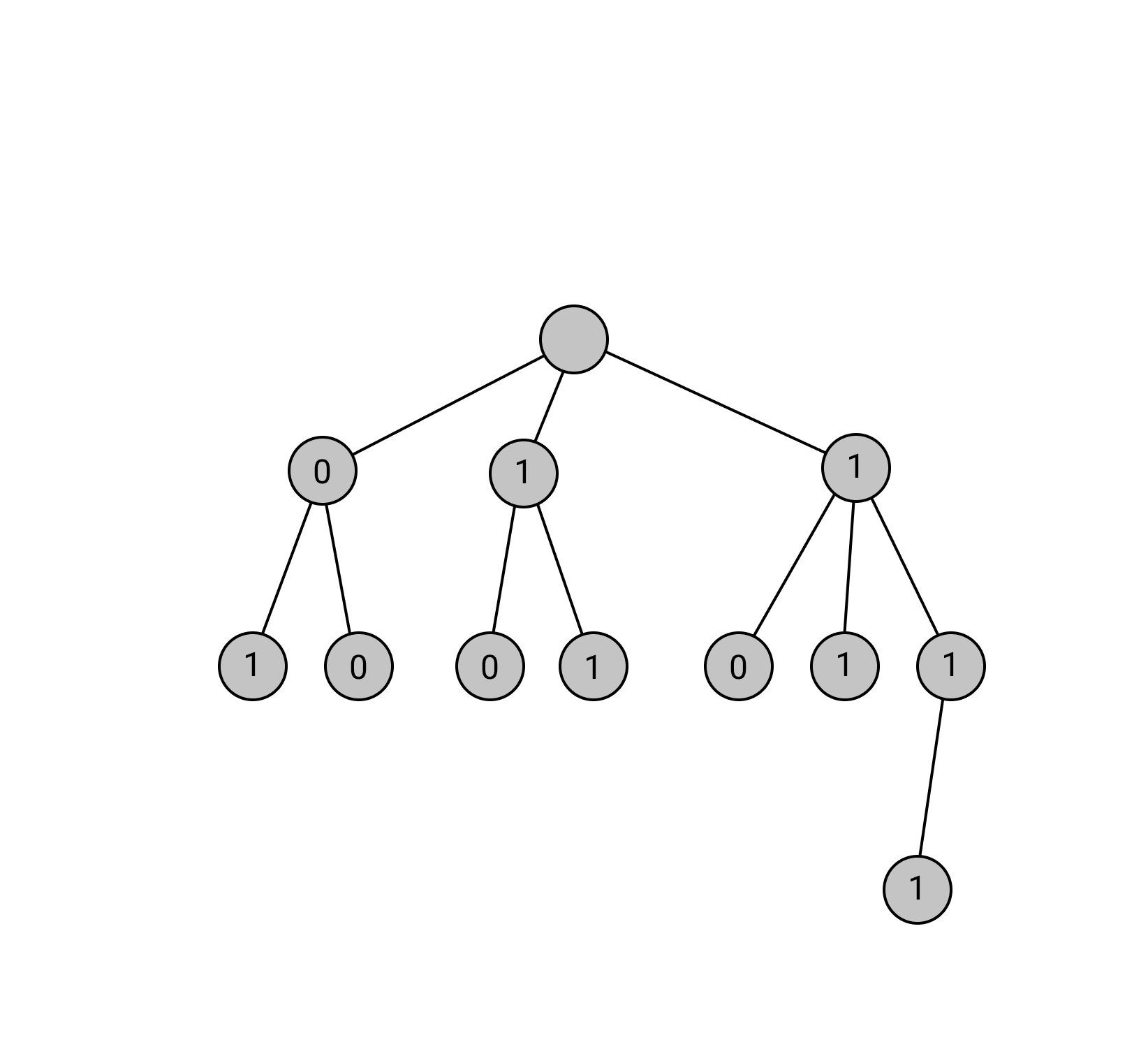}
         \caption{TED model}
         \label{fig:ted_model}
     \end{subfigure}
     %\hfill
     \begin{subfigure}[b]{0.31\textwidth}
         \centering
         \includegraphics[width=\textwidth]{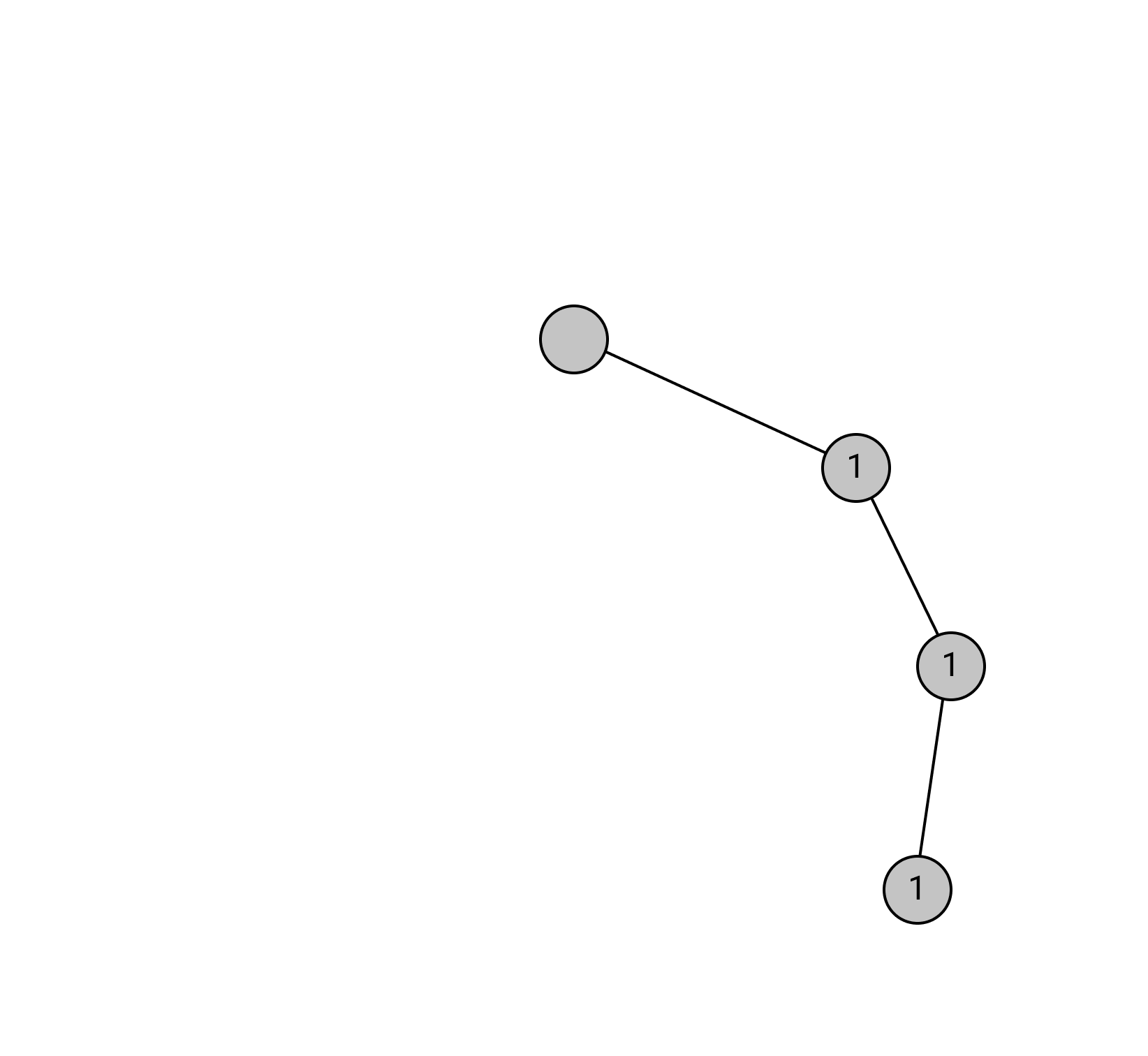}
         \caption{AON model}
         \label{fig:aon}
     \end{subfigure}
        \caption{Actions of deletion models on a sample tree. Original tree in (a), with orange nodes to be deleted. Resulting trace in the TED model (b) and the AON model (c).}
        \label{fig:del_models}
\end{figure}

Before we state our results, we recall a result of Davies, R\'acz, and Rashtchian~\cite[Theorem~4]{DRRAAP}. 

\begin{theorem}[\cite{DRRAAP}]\label{orig_thm} 
In the TED model, there exists a finite constant $C$ depending only on $q$ such that $\exp(C k\log_k n)$ traces suffice to reconstruct a complete $k$-ary tree on $n$ nodes with high probability (here $k \geq 2$).
\end{theorem}

In particular, note that the sample complexity is polynomial in $n$ whenever $k$ is a constant. 
Our first result is an alternative proof of the same result, under some mild additional assumptions, as stated below in Theorem~\ref{thm2}. 

\begin{theorem}
\label{thm2}
Fix $c \in \mathbb{Z}^+$ and let $q < \frac{c}{c+1}$.
There exists a finite constant $C$, depending only on $c$ and $q$,  
such that for any $k > c$ the following holds: in the TED model, $\exp(C k\log_k n)$ traces suffice to reconstruct a complete $k$-ary tree on $n$ nodes with high probability. 
\end{theorem} 

The additional assumptions compared to Theorem~\ref{orig_thm} are indeed mild. 
For instance, with $c = 1$ in the theorem above,  Theorem~\ref{orig_thm} is recovered for $q < 1/2$. Theorem~\ref{thm2} also allows $q$ to be arbitrarily close to $1$, provided that $k$ is at least a large enough constant. 

In \cite{DRRAAP}, the authors use combinatorial techniques to prove Theorem~\ref{orig_thm}. Our proof of Theorem~\ref{thm2} uses a mean-based complex analytic approach, similar to~\cite{DeOdonnellServedio17-WorseCase,de2019optimal,NazarovPeres17-WorseCase,KMMP19}. The advantage of our approach is that it allows us to reconstruct labels of more general tree topologies in the TED deletion model, as stated below in Theorem~\ref{tedgen}; the combinatorial proof in~\cite{DRRAAP} does not naturally lend itself to such a generalization. 

\begin{theorem}
\label{tedgen}
Let $X$ be a rooted tree on $n$ nodes with binary labels, with nodes on level $\ell$ all having the same number of children $k_{\ell}$. 
Let $k_{\max} := \max_{\ell} k_{\ell}$ and $k_{\min} := \min_{\ell} k_{\ell}$, 
where the minimum goes over all levels except the last one (containing leaf nodes). 
If $q < \frac{c}{c+1}$ and $k_{\min} > c$ for some $c \in \mathbb{Z}^+$, then there exists a finite constant $C$, depending only on $c$ and $q$, such that 
%$\exp( C k_{\max} d ) \leq \exp(Ck_{\max} \log_{k_{\min}}n)$ 
$\exp(Ck_{\max} \log_{k_{\min}}n)$ 
traces suffice to reconstruct $X$ with high probability. %, 
%where $d$ is the depth of $X$. 
\end{theorem}

Furthermore, with some slight modifications, our proof of Theorem~\ref{thm2} also provides a sample complexity bound for reconstructing arbitrary tree topologies in the AON deletion model. 

\begin{theorem}
\label{aon} 
Let $X$ be a rooted tree on $n$ nodes with binary labels, let $k_{\max}$ denote the maximum number of children a node has in $X$, and let $d$ be the depth of $X$. 
In the AON model, there exists a finite constant $C$ depending only on $q$ such that  $\exp(C k_{\max} d)$ traces suffice to reconstruct $X$ with high probability. 
\end{theorem}

The key idea in the above proofs is the notion of a \textit{subtrace}, which is the subgraph of a trace that consists only of root-to-leaf paths of length $d$, where $d$ is the depth of the underlying tree. In the proofs of Theorems~\ref{thm2} and~\ref{tedgen} we essentially only use the information contained in these subtraces and ignore the rest of the trace. 
This trick is key to making the setup amenable to the mean-based complex analytic techniques. 

The rest of the paper follows the following outline. We start with some preliminaries in Section~\ref{sec:prelims}, where we state basic tree definitions and define the notion of a subtrace more precisely. In Section~\ref{sec:thm2proof} we present our proof of Theorem~\ref{thm2}. In Section~\ref{sec:general_proofs} we generalize the methods of Section~\ref{sec:thm2proof} to a broader class of tree topologies and deletion models, proving Theorems~\ref{tedgen} and~\ref{aon}. 
%We end with some concluding remarks in Section~\ref{sec:conclusion}. 
We conclude in Section~\ref{sec:conclusion}. 

\section{Preliminaries}\label{sec:prelims}

In what follows, $X$ denotes an underlying rooted tree of known topology along with binary labels associated with the $n$ non-root nodes of the tree.

%\subsection{Basic tree terminology}

\textbf{Basic tree terminology.} 
A \textit{tree} is an acyclic graph. A rooted tree has a special node that is designated as the \textit{root}. A \textit{leaf} is a node of degree 1. We say that a node $v$ is at \textit{level} $\ell$ if the graph distance between $v$ and the root is $\ell$. We say that node $v$ is at \textit{height} $h$ if the largest graph distance from $v$ to a leaf is $h$. \textit{Depth} is the largest distance from the root to a leaf. We say that node $u$ is a \textit{child} of node $v$ if there is an edge between $u$ and $v$ and $v$ is closer to the root than~$u$ in graph distance. Similarly, we also call $v$ the \textit{parent} of $u$. More generally, $v$ is an \textit{ancestor} of~$u$ if there exists a path $v = x_0$, $\ldots$, $x_n = u$ such that $x_i$ is closer to the root than $x_{i+1}$ for every $i \in \{0,1,\ldots, n-1\}$. A \textit{complete $k$-ary tree} is a tree in which every non-leaf node has $k$ children. 

%\subsection{Subtrace augmentation}

\textbf{Subtrace augmentation.} 
Above we defined the subtrace $Z$ as the subgraph of the trace $Y$ containing all root-to-leaf paths of length $d$, where $d$ is the depth of $X$. In what follows, it will be helpful to slightly modify the definition of the subtrace by augmenting $Z$ to $Z'$ such that $Z'$ is a complete $k$-ary tree that contains $Z$ as a subgraph. Given $Z$, we construct $Z'$ recursively as follows. We begin by setting $Z' := Z$. If the root of $Z'$ currently has fewer than $k$ children, then add more child nodes to the root to the right of the existing children and label them 0. Now, consider the leftmost node in level 1 of $Z'$. If it has fewer than $k$ children, add new children to the right of the existing children of this node and label them 0. Then repeat the same procedure for the second leftmost node in level 1. Continue this procedure left to right for each level, moving from top to bottom of the tree. See Figure \ref{fig:eg_subtrace} for an illustration of this process. In Section~\ref{sec:thm2proof}, when we mention the subtrace of $X$ we mean the augmented subtrace, constructed as described here. In Section~\ref{sec:general_proofs}, we will slightly modify the notion of an augmented subtrace for the different tree topologies we will be considering.  

\begin{figure}[H]
     \centering
     \begin{subfigure}[b]{0.48\textwidth}
         \centering
         \includegraphics[width=\textwidth]{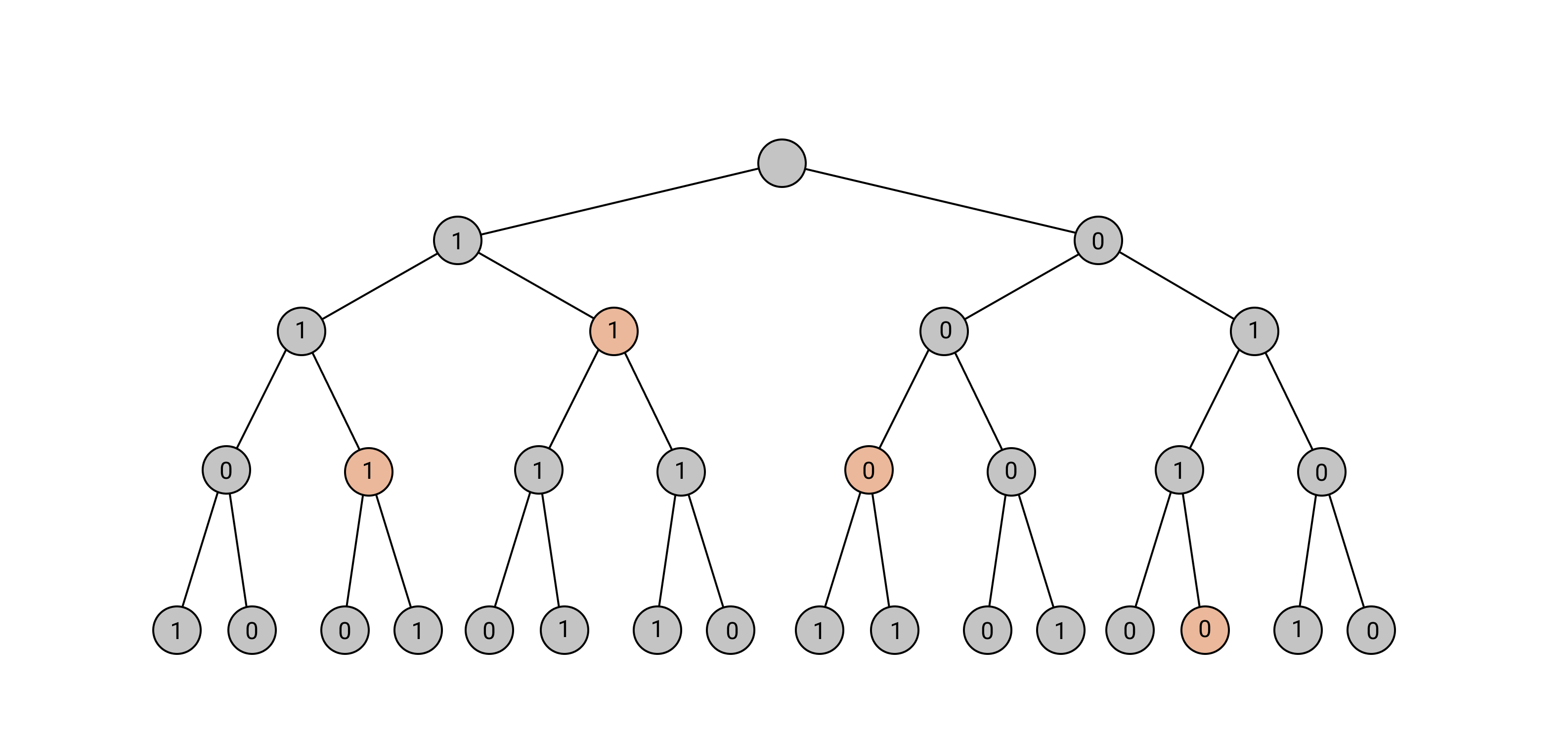}
         \caption{Original tree.}
         \label{fig:orig_big}
     \end{subfigure}
     %\hfill
     \begin{subfigure}[b]{0.48\textwidth}
         \centering
         \includegraphics[width=\textwidth]{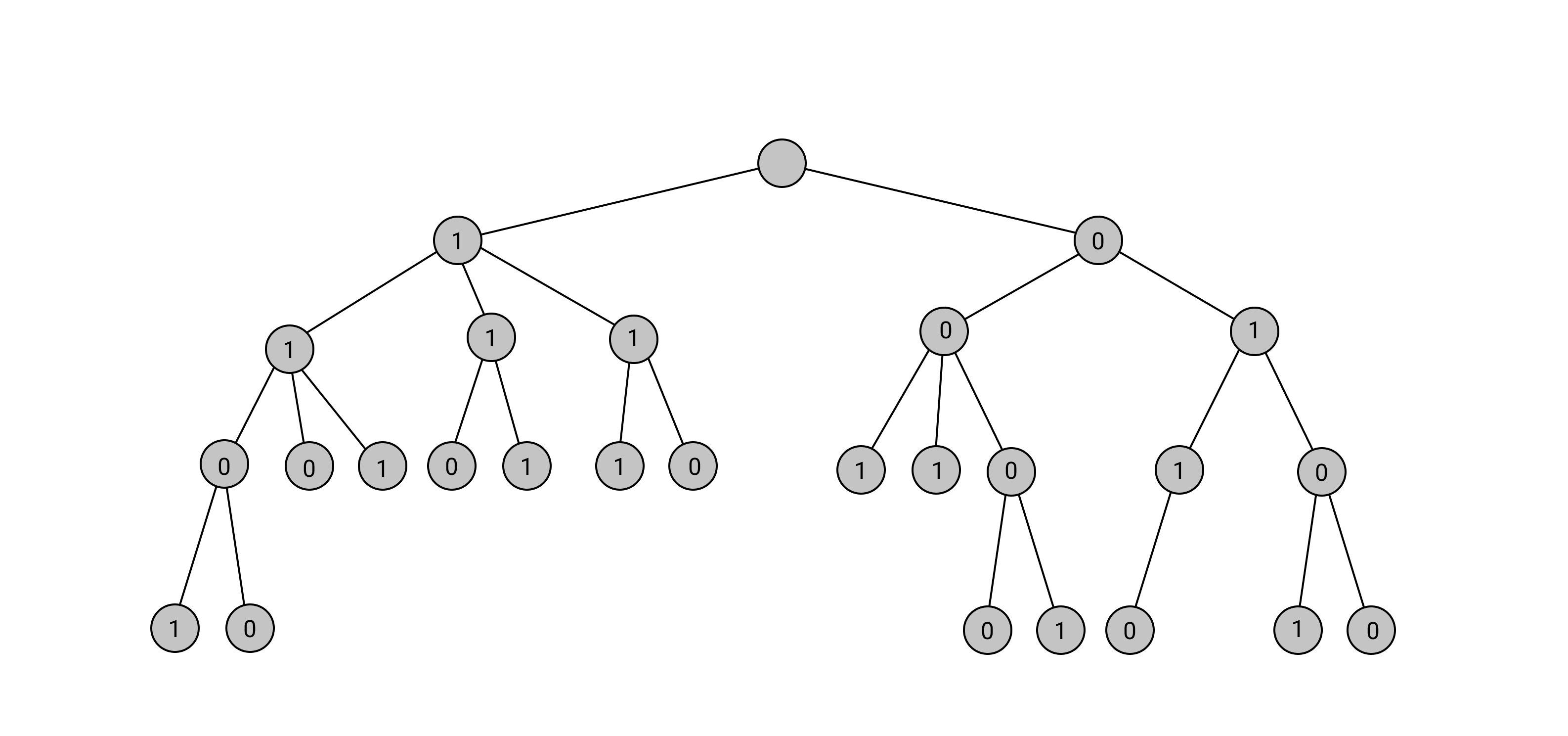}
         \caption{TED trace.}
         \label{fig:ted}
     \end{subfigure}
     %\hfill
     \begin{subfigure}[b]{0.48\textwidth}
         \centering
         \includegraphics[width=\textwidth]{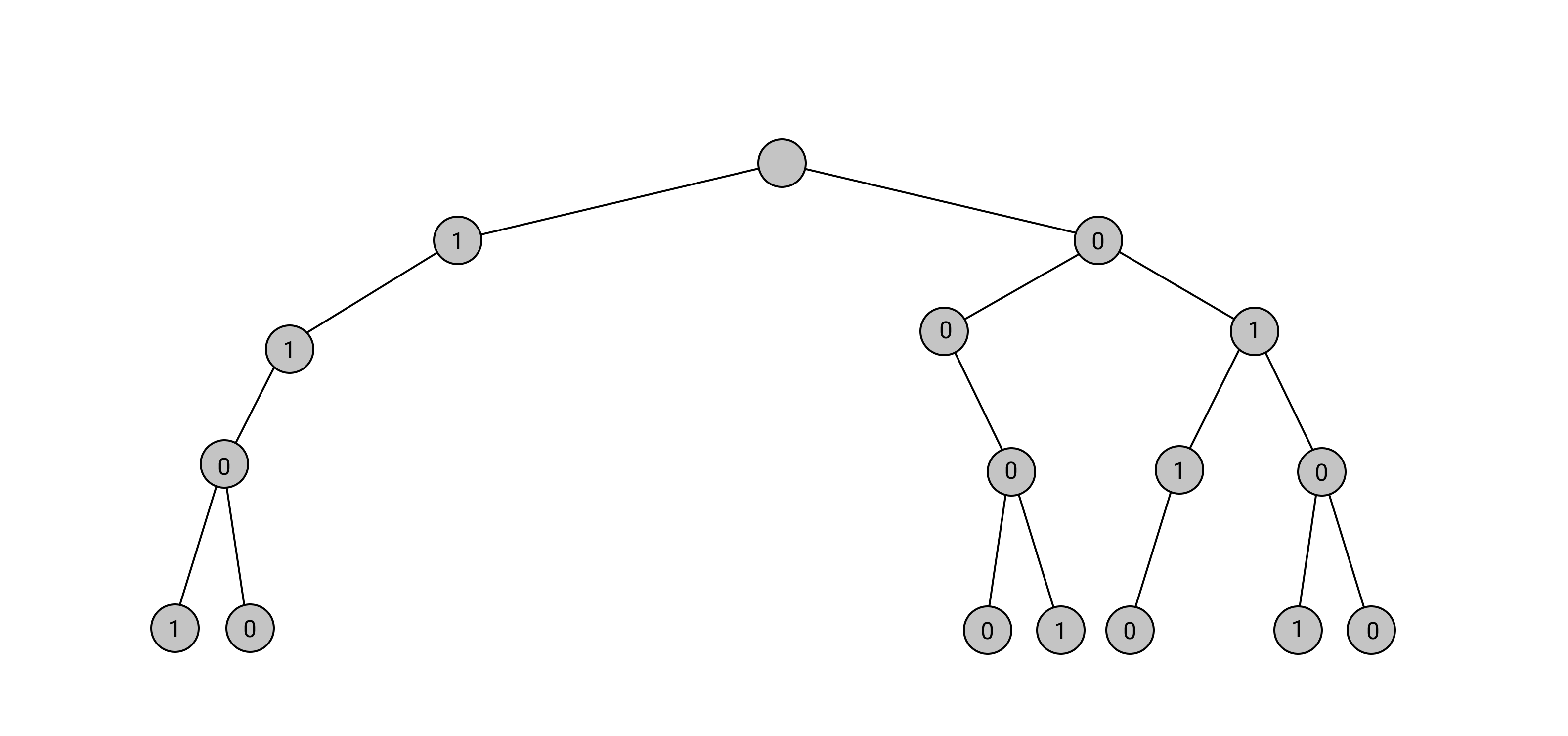}
         \caption{Subtrace.}
         \label{fig:sub}
     \end{subfigure}
     \begin{subfigure}[b]{0.48\textwidth}
         \centering
         \includegraphics[width=\textwidth]{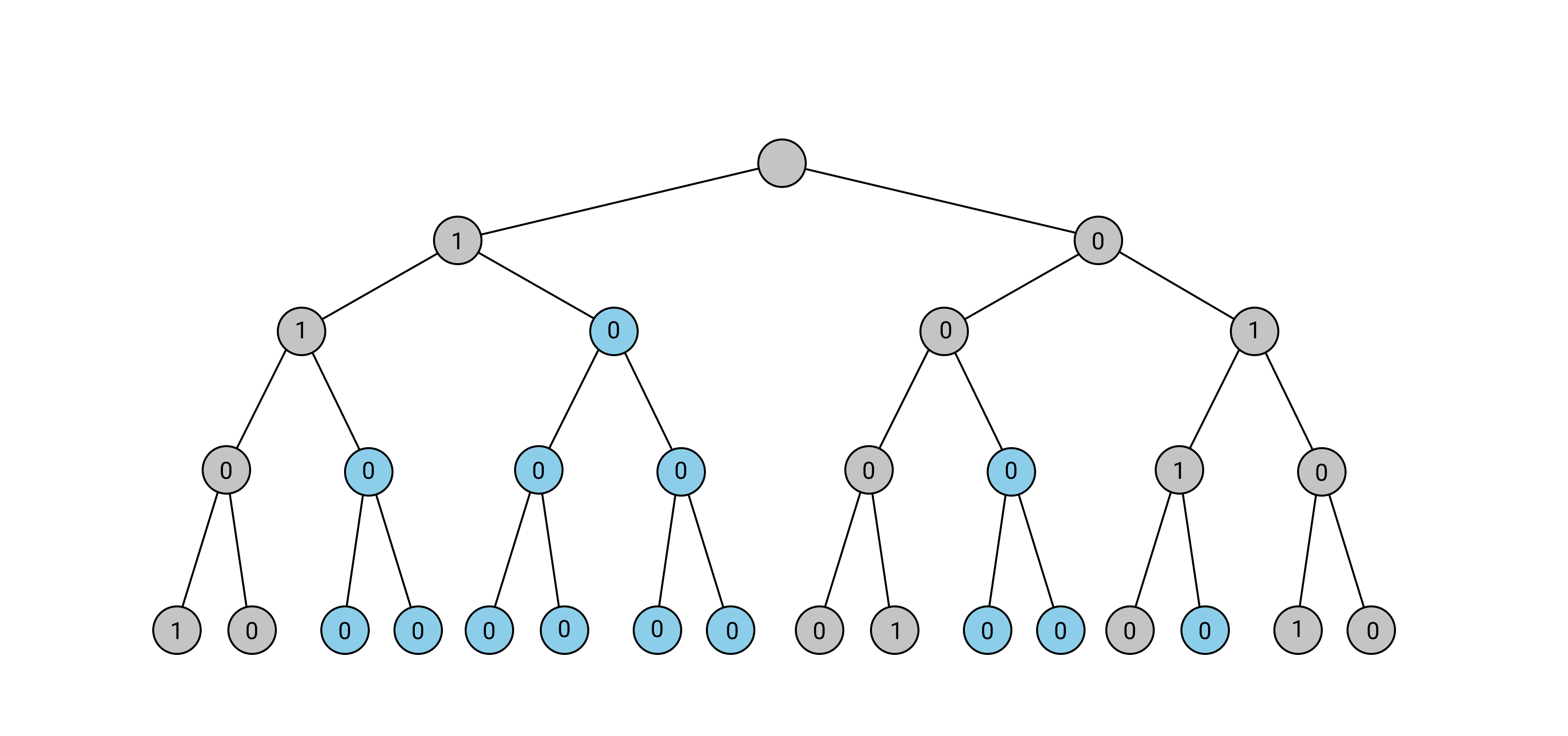}
         \caption{Augmented subtrace.}
         \label{fig:aug_sub}
     \end{subfigure}
        \caption{Construction of an augmented subtrace. Original tree in (a), with orange nodes to be deleted. The resulting trace under the TED deletion model in (b). Subtrace in (c). Augmented subtrace in (d), with blue nodes corresponding to the padding 0s.}
        \label{fig:eg_subtrace}
\end{figure}

\section{Reconstructing complete \texorpdfstring{$k$}{k}-ary trees in the TED model} \label{sec:thm2proof}

In this section we prove Theorem~\ref{thm2}. The proof takes inspiration from~\cite{DeOdonnellServedio17-WorseCase,de2019optimal,NazarovPeres17-WorseCase,KMMP19}. We begin by computing, for every node in the original tree, its probability of survival in a subtrace. We then derive a multivariate complex generating function for every level $\ell$, with random coefficients corresponding to the labels of nodes in the subtrace. Finally, we show how we can ``average'' the subtraces to determine the correct labeling for each level of the original tree with high probability. 

\subsection{Computing the probability of node survival in a subtrace} 

Let $d$ denote the depth of the original tree $X$. 
Let $Y$ denote a trace of $X$ and let $Z$ denote the corresponding subtrace obtained from $Y$. 
Observe that a node $v$ at level $\ell$ of $X$ survives in the subtrace~$Z$ if and only if a root-to-leaf path that includes $v$ survives in the trace $Y$. Furthermore, there exists exactly one path from the root to $v$, which survives in $Y$ with probability $(1-q)^{\ell-1}$ (since each of the $\ell-1$ non-root ancestors of $v$ has to survive independently). 
Let $p_{d-\ell}$ denote the probability that no $v$-to-leaf path survives in $Y$. 
Thus, 
\[
\p \left(v \text{ survives in } Z \right) = (1-q)^{\ell-1}(1- p_{d-\ell}).
\]
Thus, we can see that it suffices to compute $p_h$ for $h \in \{ 0, 1, \ldots, d-1\}$ in order to compute the probability of survival of $v$ in a subtrace. The rest of this subsection is thus dedicated to understanding $\{p_h\}_{h=0}^{d-1}$. We will not find an explicit expression for $p_h$, but rather derive a recurrence relation for $p_h$, which will prove to be good enough for us.

Let us denote by $v$ a vertex at height $h$, which is the root of the subtree under consideration. There are two events that contribute to $p_h$. Either $v$ gets deleted (this happens with probability $q$) or all of the $k$ subtrees rooted at the children of $v$ do not have a surviving root-to-leaf path in the subtrace (this happens with probability $p_{h-1}^k$). Thus, we have the following recurrence relation: for every $h \geq 0$ we have that
\begin{equation}\label{eq:rec}
  p_{h+1} = q+(1-q)p_{h}^k; 
  %= 1 - (1-q)(1-p_{h}^k);  
\end{equation}
furthermore, the initial condition satisfies 
$p_0 = q$. 
This recursion allows to compute $\{ p_{h} \}_{h=0}^{d-1}$. 
We now prove the following statement about this recursion, which will be useful later on. 

\begin{lemma}
\label{rec}
Suppose that $0 < q < \frac{c}{c+1}$ and $k > c$ for some $c \in \mathbb{Z}^+$. 
There exists $p' < 1$, depending only on $c$ and $q$, such that $p_i \leq p' < 1$ for every $i \geq 0$.  
\end{lemma}
\begin{proof}
The function  
$f(p) := 1 + p + \ldots + p^{c}$ 
is continuous and strictly increasing on $[0,1]$ 
with $f(0) = 1$ and $f(1) = c+1$. 
The assumption $q \in (0,c/(c+1))$ implies that 
$1/(1-q) \in (1,c+1)$, 
so there exists a unique $p' \in (0,1)$ such that $f(p') = 1/(1-q)$. 
By construction $p'$ is a function of $c$ and $q$. 
We will show by induction that $p_{i} \leq p'$ for every $i \geq 0$. 

First, observe that $f(p) \leq \sum_{m \geq 0} p^{m} = 1/(1-p)$. 
Thus $1/(1-p') \geq 1/(1-q)$ and so $q \leq p'$. Since $p_0 = q$, this proves the base case of the induction. 

For the induction step, first note that if $p \in (0,1)$, then 
$f(p) = (1- p^{c+1})/(1-p)$. 
Therefore the equation 
$f(p') = 1/(1-q)$ 
implies that 
$(1-q) \left( p' \right)^{c+1} = p' - q$. 
So if $p_{i} \leq p'$, then~\eqref{eq:rec} implies that 
$p_{i+1} = q + (1-q) p_i^{k} 
\leq q + (1-q)\left(p'\right)^{c+1} 
= q + (p' - q) = p'$, 
where we also used the assumption $k \geq c+1$ in the inequality. 
\end{proof}

\subsection{Generating function derivation}

We begin by introducing some additional notation. Let $b_{\ell, i}$ denote the label of the node located at level $\ell$ in position $i$ from the left in the original tree~$X$; we start the indexing of $i$ from $0$, so~$b_{\ell,0}$ is the label of the leftmost vertex on level $\ell$. 
Similarly, let $a_{\ell, i}$ denote the label of the node located at level $\ell$ in position $i$ from the left in the subtrace $Z$. 
Observe that for $i \in \{0,1, \ldots, k^{\ell} -1\}$ we can write $i = t_{\ell-1} k^{\ell-1} + t_{\ell-2} k^{\ell-2} + \ldots + t_0$ with $t_i \in \{0, \ldots, k-1\}$, that is, $t_{\ell-1}t_{\ell-2}\ldots t_0$ is the base $k$ representation of $i$. 
To abbreviate notation, we will write $a_{\ell, i} = a_{\ell, t_{\ell-1} k^{\ell-1} + t_{\ell-2} k^{\ell-2} + \ldots + t_0}$ as simply~$a_{\ell, t_{\ell-1}\ldots t_0}$.

We introduce, for every level $\ell$, a multivariate complex generating function 
whose coefficients are the labels of the nodes at level $\ell$ of a subtrace $Z$. 
Specifically, 
we introduce complex variables $w_{0}, \ldots, w_{\ell-1}$ for each position in the base $k$ representation and define 
\begin{equation}\label{eq:genfn_def}
A_{\ell}(w) := \sum_{t_{0} = 0}^{k-1} \ldots \sum_{t_{\ell-1} = 0}^{k-1} a_{\ell, t_{\ell-1}\ldots t_0} w_{\ell-1}^{t_{\ell-1}} \ldots w_0^{t_0}.
\end{equation}
We are now ready to state the main result of this subsection, which computes the expectation of this generating function. 

\begin{lemma}
\label{genfun}
For every $\ell \in \{1, \ldots, d\}$ we have that 
\begin{equation}\label{eq:genfn}
\mathbb{E}[A_{\ell}(w)] = (1-q)^{\ell-1}(1-p_{d-\ell})\sum_{t_0=0}^{k-1} \ldots \sum_{t_{\ell-1} = 0}^{k-1} b_{\ell, t_{\ell-1}\ldots t_0} \prod_{m=0}^{\ell-1}((1-p_{d-\ell+m})w_{m} + p_{d-\ell+m})^{t_m}.
\end{equation}
\end{lemma} 
This lemma is useful because the right hand side of~\eqref{eq:genfn} contains the labels of the nodes on level~$\ell$ of~$X$, while the left hand side can be estimated by averaging over subtraces. 

\begin{proof}[Proof of Lemma~\ref{genfun}] 
By linearity of expectation we have that 
\begin{equation}\label{eq:exp_lin}
\E \left[ A_{\ell}(w) \right] = \sum_{t_{0} = 0}^{k-1} \ldots \sum_{t_{\ell-1} = 0}^{k-1} \E \left[ a_{\ell, t_{\ell-1}\ldots t_0} \right] w_{\ell-1}^{t_{\ell-1}} \ldots w_0^{t_0},
\end{equation}
so our goal is to compute 
$\E \left[a_{\ell, t_{\ell-1}\ldots t_0} \right]$. 
For node $i = i_{\ell-1}\ldots i_0$ on level $\ell$, we may interpret each digit $i_m$ in the base $k$ representation as follows: 
consider node $i$'s ancestor on level $\ell - m$; the horizontal position of this node amongst its siblings is $i_m$. Thus, if the original bit $b_{\ell,i}$ survives in the subtrace, it can only end up in position $j = j_{\ell-1}j_{\ell-2}\ldots j_0$ on level $\ell$ satisfying $j_{m} \leq i_{m}$ for every $m$. 
If the $m$th digit of the location of $b_{\ell,i}$ in the subtrace is $j_{\ell-m}$, then exactly $i_{\ell-m} - j_{\ell-m}$ siblings left of the ancestor of $i$ on level $\ell-m$ must have gotten deleted and the ancestor of $i$ on level $\ell-m$ must have survived in the subtrace. Thus, the probability that bit $a_{\ell, j}$ of the subtrace is the original bit~$b_{\ell,i}$ is given by  
\begin{align*}
\p(a_{\ell, j_{\ell-1}\ldots j_0} = b_{\ell, i_{\ell-1}\ldots i_0}) 
&= \binom{i_0}{j_0} p_{d-\ell}^{i_0 - j_0} (1-p_{d-\ell})^{j_0+1} \times \binom{i_1}{j_1} p_{d-\ell+1}^{i_1 - j_1} (1-p_{d-\ell+1})^{j_1}(1-q) \times 
\ldots \\
&\quad \times \ldots 
\times \binom{i_{\ell-1}}{j_{\ell-1}} p_{d-1}^{i_{\ell-1} - j_{\ell-1}} (1-p_{d-1})^{j_{\ell-1}}(1-q) \\
&= (1-q)^{\ell-1} (1-p_{d-\ell})\prod_{m=0}^{\ell-1} \binom{i_m}{j_m} p_{d-\ell+m}^{i_m - j_m}(1-p_{d-\ell + m})^{j_m}
\end{align*} 
Summing over all $i$ satisfying $i_{m} \geq t_{m}$ for every $m$, 
and plugging into~\eqref{eq:exp_lin}, we obtain that 
\begin{align*}
\E \left[ A_{\ell}(w) \right] 
&= (1-q)^{\ell-1} (1-p_{d-\ell}) \times \\
&\quad \times \sum_{t_{0} = 0}^{k-1} \ldots \sum_{t_{\ell-1} = 0}^{k-1} 
\sum_{i_{0} = t_0}^{k-1} \ldots \sum_{i_{\ell-1} = t_{\ell-1}}^{k-1} 
b_{\ell, i_{\ell-1}\ldots i_0} 
\prod_{m=0}^{\ell-1} \binom{i_m}{t_m} p_{d-\ell+m}^{i_m - t_m}(1-p_{d-\ell + m})^{t_m} w_{m}^{t_m}.
\end{align*}
Interchanging the order of summations and using the binomial theorem ($\ell$ times) we obtain~\eqref{eq:genfn}. 
\end{proof}

\subsection{Bounding the modulus of the generating function}

Here we prove a simple lower bound on the modulus of a multivariate Littlewood polynomial. This bound will extend to the generating function computed above for appropriate choices of $w_m$. The argument presented here is inspired by the method of proof of~\cite[Lemma~4]{KMMP19}. Throughout the paper we let $\mathbb{D}$ denote the unit disc in the complex plane and let $\partial \mathbb{D}$ denote its boundary.

\begin{lemma}
\label{lwd}
Let $F(z_0, \ldots, z_{\ell-1})$ be a nonzero multivariate polynomial with monomial coefficients in $\{-1, 0, 1\}$. 
Then, 
\[
\sup_{z_0, \ldots, z_{\ell-1} \in \partial \mathbb{D}} \vert F(z_0, \ldots, z_{\ell - 1}) \vert \geq 1.  
\]
\end{lemma}

\begin{proof}
We define a sequence of polynomials $\{ F_{i} \}_{i = 0}^{\ell - 1}$ inductively as follows, 
where $F_{i}$ is a function of the variables $z_{i}, \ldots, z_{\ell-1}$.  
First, let $t_{0}$ be the smallest power of $z_0$ in a monomial of $F$ and let 
$F_{0} (z_0, \ldots, z_{\ell - 1}) := z_{0}^{-t_{0}} F(z_{0}, \ldots, z_{\ell - 1})$. By construction, $F_{0}$ has at least one monomial where $z_{0}$ does not appear. 
For $i \in \{1, \ldots, \ell - 1\}$, given $F_{i-1}$ we define $F_{i}$ as follows. 
Let $t_{i}$ be the smallest power of $z_{i}$ in a monomial of $F_{i-1}(0,z_{i}, \ldots, z_{\ell-1})$ and let 
$F_{i}(z_{i}, \ldots, z_{\ell - 1}) := z_{i}^{-t_{i}} F_{i-1}(0,z_{i}, \ldots, z_{\ell-1})$. 
Observe that this construction guarantees, for every $i$, that the polynomial 
$F_{i}(z_{i}, \ldots, z_{\ell - 1})$ 
has at least one monomial where $z_{i}$ does not appear. 
In particular, the univariate polynomial $F_{\ell-1}(z_{\ell - 1})$ has a nonzero constant term. Since the coefficients of the polynomial are in $\{-1,0,1\}$, this means that the constant term has absolute value $1$, that is, $|F_{\ell -1}(0)| = 1$.

Let $\left( z_{0}^{*}, \ldots, z_{\ell-1}^{*} \right)$ denote the maximizer of $\vert F(z_0, \ldots, z_{\ell-1}) \vert$ with $z_m \in \partial \mathbb{D}$ for all $m$. Now, by the maximum modulus principle, observe that for all $i$ we have that  
\[
\vert (z_i^*)^{t_i} \vert \vert F_i(z_i^*, \ldots, z_{\ell-1}^*) \vert =\vert F_i(z_i^*, \ldots, z_{\ell-1}^*) \vert \geq \vert F_i(0, z_{i+1}^*, \ldots, z_{\ell-1}^*) \vert.
\]
Using the definition of $F_i$ and iterating the above inequality yields, for all $i \in \{0, \ldots, \ell-2 \}$, that 
\[
\vert (z_i^*)^{t_i} \vert \vert F_i(z_i^*, \ldots, z_{\ell-1}^*) \vert \geq \vert F_{i+1}(0, z_{i+2}^*, \ldots, z_{\ell-1}^*) \vert.
\]
By taking the two ends of this chain of inequalities we can thus see that
\[
\vert F(z_0^*, \ldots, z_{\ell-1}^*) \vert \geq \vert F_{\ell-1}(0) \vert = 1. \qedhere
\]
\end{proof}

\subsection{Finishing the proof of Theorem \ref{thm2}} \label{sec:thm2proof_finish}

\begin{proof}[Proof of Theorem \ref{thm2}]
Let $X'$ and $X''$ be two complete $k$-ary trees on $n$ non-root nodes with different binary node labels. Our first goal is to distinguish between $X'$ and $X''$ using subtraces. At the end of the proof we will then explain how to estimate the original tree $X$ using subtraces. 

Since $X'$ and $X''$ have different labels, there exists at least one level of the tree where the node labels differ. Call the minimal such level $\ell_{*} = \ell_{*} \left( X', X'' \right)$; we will use this level of the subtraces to distinguish between $X'$ and $X''$. 
Let 
\[
\left\{ b'_{\ell_{*},i} : i \in \left\{ 0, 1, \ldots, k^{\ell_{*}} - 1 \right\} \right\}  
\quad
\text{ and }
\quad
\left\{ b''_{\ell_{*},i} : i \in \left\{ 0, 1, \ldots, k^{\ell_{*}} - 1 \right\} \right\}
\]
denote the labels on level $\ell_{*}$ of $X'$ and $X''$, respectively. 
Furthermore, for every $i$ define 
$b_{\ell_{*},i} := b'_{\ell_{*},i} - b''_{\ell_{*},i}$. 
By construction, $b_{\ell_{*},i} \in \{ -1, 0, 1 \}$ for every $i$,  
and there exists $i$ such that $b_{\ell_{*},i} \neq 0$. 
Let $Z'$ and $Z''$ be subtraces obtained from $X'$ and $X''$, respectively, 
and let 
\[
\left\{ Z'_{\ell_{*},i} : i \in \left\{ 0, 1, \ldots, k^{\ell_{*}} - 1 \right\} \right\} 
\quad 
\text{ and }
\quad 
\left\{ Z''_{\ell_{*},i} : i \in \left\{ 0, 1, \ldots, k^{\ell_{*}} - 1 \right\} \right\}
\]
denote the labels on level $\ell_{*}$ of $Z'$ and $Z''$, respectively. 
By Lemma~\ref{genfun} we have that 
\begin{multline*}
\E \left[ \sum_{t_{0} = 0}^{k-1} \ldots \sum_{t_{\ell_{*}-1} = 0}^{k-1} Z'_{\ell_{*}, t_{\ell_{*}-1}\ldots t_0} \prod_{m=0}^{\ell_{*}-1} w_{m}^{t_{m}} \right] 
- 
\E \left[ \sum_{t_{0} = 0}^{k-1} \ldots \sum_{t_{\ell_{*}-1} = 0}^{k-1} Z''_{\ell_{*}, t_{\ell_{*}-1}\ldots t_0} \prod_{m=0}^{t_{\ell_{*}-1}} w_{m}^{t_{m}} \right] \\
= (1-q)^{\ell_{*}-1}(1-p_{d-\ell_{*}}) \sum_{t_0=0}^{k-1} \ldots \sum_{t_{\ell_{*}-1} = 0}^{k-1} b_{\ell_{*}, t_{\ell_{*}-1}\ldots t_0} \prod_{m=0}^{\ell_{*}-1}((1-p_{d-\ell_{*}+m})w_{m} + p_{d-\ell_{*}+m})^{t_m}.
\end{multline*}
Now define the multivariate polynomial $B(z)$ in the variables $z=\left(z_{0}, \ldots, z_{\ell_{*}-1} \right)$ as follows: 
\[
B(z) := \sum_{t_0=0}^{k-1} \ldots \sum_{t_{\ell_{*}-1} = 0}^{k-1} b_{\ell_{*}, t_{\ell_{*}-1}\ldots t_0} \prod_{m=0}^{\ell_{*}-1} z_{m}^{t_m}.
\]
Lemma~\ref{lwd} implies that there exists 
$z^{*} = \left(z_{0}^{*}, \ldots, z_{\ell_{*}-1}^{*} \right)$ 
such that $z_{m}^{*} \in \partial \mathbb{D}$ for every $m \in \left\{ 0, \ldots, \ell_{*} - 1 \right\}$ and 
\[
\left| B \left( z^{*} \right) \right| \geq 1. 
\]
For $m \in \left\{ 0, \ldots, \ell_{*} - 1 \right\}$ let 
\[
w_{m}^{*} := \frac{z_{m}^{*} - p_{d-\ell_{*}+m}}{1 - p_{d-\ell_{*}+m}}. 
\]
Note that the polynomial $B$ is a function of $X'$ and $X''$, and thus so is $z^{*}$ and also $w^{*} = \left(w_{0}^{*}, \ldots, w_{\ell_{*}-1}^{*} \right)$. 
Putting together the four previous displays and using the triangle inequality we obtain that 
\begin{equation}\label{eq:mod_bound}
\sum_{t_{0} = 0}^{k-1} \ldots \sum_{t_{\ell_{*}-1} = 0}^{k-1} 
\left| \E \left[ Z'_{\ell_{*}, t_{\ell_{*}-1}\ldots t_0} \right] - \E \left[ Z''_{\ell_{*}, t_{\ell_{*}-1}\ldots t_0} \right] \right| 
\prod_{m=0}^{\ell_{*}-1} \left| w_{m}^{*} \right|^{t_{m}} 
\geq 
(1-q)^{\ell_{*}-1}(1-p_{d-\ell_{*}}).
\end{equation}
Next we estimate 
$\left| w_{m}^{*} \right|$. 
By the definition of $w_{m}^{*}$ and the triangle inequality we have that 
\begin{equation}\label{eq:w_magnitude}
\left| w_{m}^{*} \right| 
= \frac{\left| z_{m}^{*} - p_{d-\ell_{*}+m} \right|}{1 - p_{d-\ell_{*}+m}} 
\leq \frac{\left| z_{m}^{*} \right| + p_{d-\ell_{*}+m}}{1 - p_{d-\ell_{*}+m}}
\leq \frac{2}{1 - p'},
\end{equation}
where in the last inequality we used that 
$\left| z_{m}^{*} \right| = 1$ 
and that 
$p_{d-\ell_{*}+m} \leq p' < 1$ (from Lemma~\ref{rec}). 
Note that $p'$ is a constant that depends only on $c$ and $q$ (recall that $c$ is an input to the theorem). 
The bound in~\eqref{eq:w_magnitude} implies that 
\begin{equation}\label{eq:prod_w_to_t}
\prod_{m=0}^{\ell_{*}-1} \left| w_{m}^{*} \right|^{t_{m}} 
\leq \left( \frac{2}{1 - p'} \right)^{k \ell_{*}}.
\end{equation}
Plugging this back into~\eqref{eq:mod_bound} (and using that $p_{d-\ell_{*}} \leq p'$) we get that 
\[
\sum_{t_{0} = 0}^{k-1} \ldots \sum_{t_{\ell_{*}-1} = 0}^{k-1} 
\left| \E \left[ Z'_{\ell_{*}, t_{\ell_{*}-1}\ldots t_0} \right] - \E \left[ Z''_{\ell_{*}, t_{\ell_{*}-1}\ldots t_0} \right] \right| 
\geq 
(1-q)^{\ell_{*}-1}(1-p')^{k \ell_{*} + 1} (1/2)^{k\ell_{*}}.
\]
Thus by the pigeonhole principle there exists $i_{*} \in \left\{ 0, 1, \ldots, k^{\ell_{*}} - 1 \right\}$ such that 
\begin{equation}\label{eq:gap}
\left| \E \left[ Z'_{\ell_{*}, i_{*}} \right] - \E \left[ Z''_{\ell_{*}, i_{*}} \right] \right| 
\geq 
\frac{(1-q)^{\ell_{*}-1}(1-p')^{k \ell_{*} + 1} (1/2)^{k\ell_{*}}}{k^{\ell_{*}}} 
\geq 
\exp \left( - C k \ell_{*} \right) 
\geq 
\exp \left( - C k \log_{k} n \right),
\end{equation}
where the second inequality holds for a large enough constant $C$ that depends only on $c$ and $q$, 
while the third inequality is because the depth of the tree is $\log_{k} n$. 
Note that $i_{*}$ is a function of~$X'$ and $X''$. 

Now suppose that we sample $T$ traces of $X$ from the TED deletion channel and let $Z^{1}, \ldots, Z^{T}$ denote the corresponding subtraces. 
Let $X'$ and $X''$ be two complete $k$-ary labeled trees with different labels, and recall the definitions of $\ell_{*} = \ell_{*} \left( X', X'' \right)$ and $i_{*} = i_{*} \left( X', X'' \right)$ from above. 
We say that $X'$ \emph{beats} $X''$ (with respect to these samples) if 
\[
\left| \frac{1}{T} \sum_{t=1}^{T} Z^{t}_{\ell_{*},i_{*}} - \E \left[ Z'_{\ell_{*}, i_{*}} \right] \right| 
< 
\left| \frac{1}{T} \sum_{t=1}^{T} Z^{t}_{\ell_{*},i_{*}} - \E \left[ Z''_{\ell_{*}, i_{*}} \right] \right|.
\]
We are now ready to define our estimate $\widehat{X}$ of the labels of the original tree. 
If there exists a complete $k$-ary tree $X'$ that beats every other complete $k$-ary tree (with respect to these samples), then we let $\widehat{X} := X'$. Otherwise, define $\widehat{X}$ arbitrarily. 

Finally, we show that this estimate is correct with high probability. 
Let $\eta := \exp \left( - C k \log_{k} n \right)$. 
By a union bound and a Chernoff bound (using~\eqref{eq:gap}), the probability that the estimate is incorrect is bounded by 
\[
\p \left( \widehat{X} \neq X \right) 
\leq 
\sum_{X' : X' \neq X} \p \left( X' \text{ beats } X \right) 
\leq 2^{n} \exp \left( - T \eta^{2} / 2 \right) 
= 2^{n} \exp \left( - \frac{T}{2} \exp \left( - 2 C k \log_{k} n \right) \right).
\]
Choosing $T = \exp \left( 3 C k \log_{k} n \right)$, the right hand side of the display above tends to $0$. 
\end{proof}

\section{Reconstructing more general tree topologies} \label{sec:general_proofs}

The method of proof shown in the previous section naturally lends itself to the more general results of Theorem~\ref{tedgen} for the TED deletion model and Theorem~\ref{aon} for the AON deletion model. The proofs are almost entirely identical to the one presented above, so we will only highlight the new ideas below and leave the details to the reader. 

\subsection{TED deletion model, more general tree topologies}

Before we proceed with the proof, we must clarify the notion of a subtrace. In Section~\ref{sec:prelims} we described the notion of an augmented subtrace for a $k$-ary tree. 
More generally, for trees in the setting of Theorem~\ref{tedgen}, 
we define an augmented subtrace in a similar way; the key point is that the underlying tree structure of the augmented subtrace is the same as the underlying tree structure of $X$. 
That is, 
we start with the root of the subtrace $Z$, and if it has less than $k_{0}$ children, we add nodes with $0$ labels to the right of its existing children, until the root has $k_0$ children in total. We then move on to the leftmost node on level $1$ and add new children with label 0 to the right of its existing children, until it has $k_1$ children. We continue in this fashion from left to right on each level, ensuring that each node on level $\ell$ has $k_{\ell}$ children, moving from top to bottom of the tree.

\begin{proof}[Proof of Theorem \ref{tedgen}]
As before, we begin by computing, for every node in the tree, its probability of survival in a subtrace. 
The quantities $\{ p_h \}_{h=0}^{d-1}$ can be defined exactly as before, where again $d$ denotes the depth of the tree. 
The recurrence relation changes slightly: for every $h \geq 0$ we have that 
\[
p_{h+1} = q + (1-q) p_{h}^{k_{d-h-1}};
\]
furthermore, the initial condition satisfies $p_{0} = q$. 
The following lemma is the analog of Lemma~\ref{rec}; we omit its proof, since it is identical to that of Lemma~\ref{rec}. 
\begin{lemma}
Suppose that $0 < q < \frac{c}{c+1}$ and $k_{\min} > c$ for some $c \in \mathbb{Z}^+$. 
There exists $p' < 1$, depending only on $c$ and $q$, such that 
$p_i \leq p' < 1$ for every $i \geq 0$. 
\end{lemma}

Next, we turn to defining and analyzing an appropriate generating function. 
Note that there are $\prod_{m=0}^{\ell-1} k_{m}$ nodes on level $\ell$ of the tree $X$. Observe that every 
$i \in \left\{ 0, 1, \ldots, \prod_{m=0}^{\ell-1} k_{m} - 1 \right\}$ 
can be uniquely written as 
\begin{equation}\label{eq:rep_gen}
i = i_{\ell-1} \prod_{m=1}^{\ell-1} k_{m} + i_{\ell-2} \prod_{m=2}^{\ell-1} k_{m} + \ldots + i_1 k_{\ell-1} + i_{0},
\end{equation}
where 
$i_{m} \in \left\{ 0, 1, \ldots, k_{\ell-1-m} -1 \right\}$ 
for every $m \in \{ 0,1, \ldots, \ell-1\}$. 
The interpretation of each digit~$i_{m}$ in this representation is the same as before: 
consider node $i$'s ancestor on level $\ell - m$; the horizontal position of this node amongst its siblings is $i_{m}$. 
To abbreviate notation, we write 
$i = i_{\ell-1} \ldots i_{0}$ 
for the expression in~\eqref{eq:rep_gen}. 
With this representation of the nodes at level $\ell$, we may define the generating function for level $\ell$ as follows: 
\[
A_{\ell}(w) := \sum_{t_{0} = 0}^{k_{\ell-1}-1} \ldots \sum_{t_{\ell-1} = 0}^{k_{0}-1} a_{\ell, t_{\ell-1}\ldots t_0} w_{\ell-1}^{t_{\ell-1}} \ldots w_0^{t_0}.
\]
The following lemma is the analog of Lemma~\ref{genfun}; we omit its proof, since it is analogous to that of Lemma~\ref{genfun}. 
\begin{lemma}
\label{genfun_gen}
For every $\ell \in \{1, \ldots, d\}$ we have that 
\begin{equation*}\label{eq:genfn_gen}
\mathbb{E}[A_{\ell}(w)] = (1-q)^{\ell-1}(1-p_{d-\ell})\sum_{t_0=0}^{k_{\ell-1}-1} \ldots \sum_{t_{\ell-1} = 0}^{k_{0}-1} b_{\ell, t_{\ell-1}\ldots t_0} \prod_{m=0}^{\ell-1}((1-p_{d-\ell+m})w_{m} + p_{d-\ell+m})^{t_m}.
\end{equation*}
\end{lemma}
With these tools in place, the remainder of the proof is almost identical to Section~\ref{sec:thm2proof_finish}. 
The inequality~\eqref{eq:prod_w_to_t} is now replaced with 
\begin{equation*}\label{eq:prod_w_to_t_gen}
\prod_{m=0}^{\ell_{*}-1} \left| w_{m}^{*} \right|^{t_{m}} 
\leq \left( \frac{2}{1 - p'} \right)^{k_{\max} \ell_{*}}.
\end{equation*}
Subsequently, by the pigeonhole principle there exists 
$i_{*} \in \left\{0,1,\ldots, \prod_{m=0}^{\ell_{*}-1} k_{m}-1 \right\}$ 
such that 
\[
\left| \E \left[ Z'_{\ell_{*},i_{*}} \right] - \E \left[ Z''_{\ell_{*},i_{*}} \right] \right| 
\geq 
\exp \left( - C k_{\max} \ell_{*} \right) 
\geq 
\exp \left( - C k_{\max} d \right),
\]
where the first inequality holds for a large enough constant $C$ that depends only on $c$ and $q$. 
The rest of the proof is identical to Section~\ref{sec:thm2proof_finish}, 
showing that $T = \exp \left( 3 C k_{\max} d \right)$ traces suffice. 
The claim follows because the depth of the tree is at most $\log_{k_{\min}} n$. 
\end{proof}

\subsection{AON deletion model, arbitrary tree topologies}

We begin by first proving Theorem~\ref{aon} for complete $k$-ary trees. We will then generalize to arbitrary tree topologies. 
Importantly, in the AON model we will work directly with the tree traces, as opposed to the subtraces as we did previously. As described in Section~\ref{sec:prelims}, we augment each trace~$Y$ with additional nodes with 0 labels to form a $k$-ary tree. In what follows, when we say ``trace'' we mean this augmented trace. 

\begin{theorem}
\label{aon_kary}
In the AON model, there exists a finite constant $C$ depending only on $q$ such that $\exp(Ck\log_k n)$ traces suffice to reconstruct a complete $k$-ary tree on $n$ nodes w.h.p.\ (here $k \geq 2$). 
\end{theorem}

\begin{proof} 
We may define $A_{\ell}(w)$, the generating function for level $\ell$, exactly as in~\eqref{eq:genfn_def}. 
The following lemma is the analog of Lemma~\ref{genfun}; we omit its proof, since it is analogous to that of Lemma~\ref{genfun}. 
\begin{lemma}
\label{aon_genfun}
For every $\ell \in \{1, \ldots, d \}$ we have that 
\[
\E \left[A_{\ell}(w) \right] = (1-q)^{\ell}\sum_{t_0=0}^{k-1} \ldots \sum_{t_{\ell-1} = 0}^{k-1} b_{\ell, t_{\ell-1}\ldots t_0} \prod_{m=0}^{\ell-1}((1-q)w_{m} + q)^{t_m}.
\]
\end{lemma}

With this lemma in place, the remainder of the proof is almost identical to Section~\ref{sec:thm2proof_finish}. 
The polynomial $B(z)$, and hence also $z^{*}$, are as before. Now, we define 
$w_{m}^{*} := \left( z_{m}^{*} - q \right) / (1 - q)$. 
The right hand side of~\eqref{eq:mod_bound} becomes $(1-q)^{\ell_{*}}$. 
The analog of~\eqref{eq:w_magnitude} becomes the inequality $\left| w_{m}^{*} \right| \leq 2/(1-q)$; moreover, wherever $p'$ appears in Section~\ref{sec:thm2proof_finish}, it is replaced by $q$ here. 
Altogether, we obtain that there exists 
$i_{*} \in \left\{0,1,\ldots, k^{\ell_{*}}-1 \right\}$ 
such that 
\[
\left| \E \left[ Z'_{\ell_{*},i_{*}} \right] - \E \left[ Z''_{\ell_{*},i_{*}} \right] \right| 
\geq 
\exp \left( - C k \ell_{*} \right) 
\geq 
\exp \left( - C k d \right) 
\geq 
\exp \left( - C k \log_{k} n \right),
\]
where the first inequality holds for a large enough constant $C$ that depends only on $q$. 
The rest of the proof is identical to Section~\ref{sec:thm2proof_finish}, 
showing that $T = \exp \left( 3 C k d \right) = \exp \left( 3 C k \log_k n \right)$ traces suffice. 
\end{proof}

\begin{proof}[Proof of Theorem \ref{aon}] 
Suppose that $X$ is a rooted tree with arbitrary topology and let $k_{\max}$ denote the largest number of children a node in $X$ has. Once we sample a trace $Y$ from $X$, we form an augmented trace similarly to how we do it when $X$ is a $k$-ary tree, except now we add nodes with~$0$ labels to ensure that each node has $k_{\max}$ children. Thus, each augmented trace is a complete $k_{\max}$-ary tree. Now, let $X'$ denote a $k_{\max}$-ary tree obtained by augmenting $X$ to a $k_{\max}$-ary tree in the same fashion that we augment traces of $X$ to a $k_{\max}$-ary tree. 

As before, for each node $i$ on level $\ell$ of $X$, there is a unique representation $i = i_{\ell-1} \ldots i_{0}$ where $i_{m}$ is the position of node $i$'s ancestor on level $\ell - m$ among its siblings. 
Importantly, for every node in $X$, its representation in $X'$ is the same. 
This fact, together with the augmentation construction, implies 
that $\mathbb{E}[a_{\ell, t_{\ell-1}\ldots t_0}]$ for the node $a_{\ell, t_{\ell-1}\ldots t_0}$ in $Y$ is identical to $\mathbb{E}[a_{\ell, t_{\ell-1}\ldots t_0}]$ for the node $a_{\ell, t_{\ell-1}\ldots t_0}$ in $Y'$, which is a trace sampled from $X'$. 
Therefore, we can use the procedure presented in Theorem~\ref{aon_kary} to reconstruct $X'$ w.h.p.\ using $T = \exp \left( C k_{\max} d \right)$ traces sampled from $X$. 
By taking the appropriate subgraph of $X'$, we can thus reconstruct $X$ as well.
\end{proof}

\section{Conclusion}\label{sec:conclusion}

In this work we introduce the notion of a subtrace and demonstrate its utility in analyzing traces produced by the deletion channel in the tree trace reconstruction problem. We provide a novel algorithm for the reconstruction of complete $k$-ary trees, which matches the sample complexity of the combinatorial approach of~\cite{DRRAAP}, 
by applying mean-based complex analytic tools to the subtrace. 
This technique also allows us to reconstruct trees with more general topologies in the TED deletion model, specifically trees where the nodes at every level have the same number of children (with this number varying across levels). 

However, many questions remain unanswered; we hope that the ideas introduced here will help address them. In particular, how can we reconstruct, under the TED deletion model, arbitrary trees where all leaves are on the same level? Since the notion of a subtrace is well-defined for such trees, we hope that the proof technique presented here can somehow be generalized to answer this question.

%\clearpage

%%%%%%%%%%%%%%%%%%
%%% References %%%
%%%%%%%%%%%%%%%%%%

\bibliographystyle{plain}
\bibliography{ref.bib}

\begin{thebibliography}{10}

\bibitem{Ban}
Frank Ban, Xi~Chen, Adam Freilich, Rocco~A. Servedio, and Sandip Sinha.
\newblock Beyond trace reconstruction: Population recovery from the deletion
  channel.
\newblock In {\em 60th {IEEE} Annual Symposium on Foundations of Computer
  Science (FOCS)}, pages 745--768, 2019.

\bibitem{BCSS19}
Frank Ban, Xi~Chen, Rocco~A. Servedio, and Sandip Sinha.
\newblock Efficient average-case population recovery in the presence of
  insertions and deletions.
\newblock In {\em Approximation, Randomization, and Combinatorial Optimization.
  Algorithms and Techniques (APPROX/RANDOM)}, volume 145 of {\em LIPIcs}, pages
  44:1--44:18. Schloss Dagstuhl - Leibniz-Zentrum f{\"{u}}r Informatik, 2019.

\bibitem{BatuKannan04-RandomCase}
Tugkan Batu, Sampath Kannan, Sanjeev Khanna, and Andrew McGregor.
\newblock Reconstructing strings from random traces.
\newblock In {\em Proceedings of the Fifteenth Annual {ACM-SIAM} Symposium on
  Discrete Algorithms ({SODA})}, pages 910--918, 2004.

\bibitem{BLS20}
Joshua Brakensiek, Ray Li, and Bruce Spang.
\newblock Coded trace reconstruction in a constant number of traces.
\newblock In {\em Proceedings of the {IEEE} Annual Symposium on Foundations of
  Computer Science {(FOCS)}}, 2020.

\bibitem{chase2020ub}
Zachary Chase.
\newblock New upper bounds for trace reconstruction.
\newblock Preprint available at \url{https://arxiv.org/abs/2009.03296}, 2020.

\bibitem{LB2}
Zachary Chase.
\newblock {New Lower Bounds for Trace Reconstruction}.
\newblock {\em Annales de l’Institut Henri Poincar\'{e}, Probabilit{\'e}s et
  Statistiques}, to appear, 2021.

\bibitem{chen2020polynomialtime}
Xi~Chen, Anindya De, Chin~Ho Lee, Rocco~A Servedio, and Sandip Sinha.
\newblock Polynomial-time trace reconstruction in the smoothed complexity
  model.
\newblock In {\em Proceedings of the Annual {ACM-SIAM} Symposium on Discrete
  Algorithms ({SODA})}, 2021.

\bibitem{CGMR}
Mahdi Cheraghchi, Ryan Gabrys, Olgica Milenkovic, and Joao Ribeiro.
\newblock Coded trace reconstruction.
\newblock {\em IEEE Transactions on Information Theory}, 66(10):6084--6103,
  2020.

\bibitem{DRRAAP}
Sami Davies, Mikl\'os~Z. R\'acz, and Cyrus Rashtchian.
\newblock {Reconstructing Trees from Traces}.
\newblock {\em The Annals of Applied Probability}, to appear, 2021.

\bibitem{davies2020approximate}
Sami Davies, Mikl\'os~Z. R\'acz, Cyrus Rashtchian, and Benjamin~G. Schiffer.
\newblock Approximate trace reconstruction.
\newblock Preprint available at \url{https://arxiv.org/abs/2012.06713}, 2020.

\bibitem{DeOdonnellServedio17-WorseCase}
Anindya De, Ryan O'Donnell, and Rocco~A. Servedio.
\newblock Optimal mean-based algorithms for trace reconstruction.
\newblock In {\em Proceedings of the 49th Annual {ACM} {SIGACT} Symposium on
  Theory of Computing ({STOC})}, pages 1047--1056, 2017.

\bibitem{de2019optimal}
Anindya De, Ryan O'Donnell, and Rocco~A. Servedio.
\newblock Optimal mean-based algorithms for trace reconstruction.
\newblock {\em The Annals of Applied Probability}, 29(2):851--874, 2019.

\bibitem{grigorescu2020limitations}
Elena Grigorescu, Madhu Sudan, and Minshen Zhu.
\newblock {Limitations of Mean-Based Algorithms for Trace Reconstruction at
  Small Distance}.
\newblock Preprint available at \url{https://arxiv.org/abs/2011.13737}, 2020.

\bibitem{HartungHP18}
Lisa Hartung, Nina Holden, and Yuval Peres.
\newblock Trace reconstruction with varying deletion probabilities.
\newblock In {\em Proceedings of the Fifteenth Workshop on Analytic
  Algorithmics and Combinatorics ({ANALCO})}, pages 54--61, 2018.

\bibitem{LB1}
Nina Holden and Russell Lyons.
\newblock Lower bounds for trace reconstruction.
\newblock {\em Annals of Applied Probability}, 30(2):503--525, 2020.

\bibitem{holden2020subpolynomial}
Nina Holden, Robin Pemantle, Yuval Peres, and Alex Zhai.
\newblock Subpolynomial trace reconstruction for random strings and arbitrary
  deletion probability.
\newblock {\em Mathematical Statistics and Learning}, 2(3):275--309, 2020.

\bibitem{HolensteinMPW08}
Thomas Holenstein, Michael Mitzenmacher, Rina Panigrahy, and Udi Wieder.
\newblock Trace reconstruction with constant deletion probability and related
  results.
\newblock In {\em Proc. 19th {ACM-SIAM} Symposium on Discrete Algorithms
  ({SODA})}, pages 389--398, 2008.

\bibitem{KMMP19}
Akshay Krishnamurthy, Arya Mazumdar, Andrew McGregor, and Soumyabrata Pal.
\newblock {Trace Reconstruction: Generalized and Parameterized}.
\newblock In {\em 27th Annual European Symposium on Algorithms (ESA 2019)},
  pages 68:1--68:25, 2019.

\bibitem{levenshtein2001efficient}
Vladimir~I Levenshtein.
\newblock Efficient reconstruction of sequences.
\newblock {\em IEEE Transactions on Information Theory}, 47(1):2--22, 2001.

\bibitem{maranzatto_thesis}
Thomas~J. Maranzatto.
\newblock {Tree Trace Reconstruction: Some Results}.
\newblock Thesis, New College of Florida, 2020.

\bibitem{narayanan2021population}
Shyam Narayanan.
\newblock {Population Recovery from the Deletion Channel: Nearly Matching Trace
  Reconstruction Bounds}.
\newblock In {\em Proceedings of the {ACM-SIAM} Symposium on Discrete
  Algorithms ({SODA})}, 2021.

\bibitem{narayanan2020circular}
Shyam Narayanan and Michael Ren.
\newblock {Circular Trace Reconstruction}.
\newblock In {\em Proceedings of Innovations in Theoretical Computer Science
  (ITCS)}, 2021.

\bibitem{NazarovPeres17-WorseCase}
Fedor Nazarov and Yuval Peres.
\newblock {Trace reconstruction with $\exp(O(n^{1/3}))$ samples}.
\newblock In {\em Proceedings of the 49th Annual {ACM} {SIGACT} Symposium on
  Theory of Computing ({STOC})}, pages 1042--1046, 2017.

\end{thebibliography}

\end{document}